\documentclass{article}

\usepackage{graphicx}

\usepackage{latexsym,cmmib57,amssymb}
\usepackage{amsmath,amsthm}
\usepackage[psamsfonts]{eucal}
\usepackage{url}

\newtheorem{theorem}{Theorem}
\newtheorem{definition}{Definition}
\newtheorem{corollary}[theorem]{Corollary}
\newtheorem{lemma}[theorem]{Lemma}

\newcommand{\condgen}[6]{{#1}#2 #5 #3 #6 #4}
\newcommand{\bbrd}[1]{\mbox{\rm{I}\kern-.1667em{#1}}}
\newcommand{\EXP}{\mathbb{E}}
\newcommand{\PROB}{\mathbb{P}}

\newcommand{\Probcmd}[2]{\condgen{\PROB}{\Bigl\{}{\Bigm|}{\Bigr\}}{#1}{#2}}

\newcommand{\Expcmd}[2]{\condgen{\EXP}{\Bigl[}{\Bigm|}{\Bigr]}{#1}{#2}}

\DeclareMathOperator{\Var}{Var}

\newsavebox{\fmbox}
\newenvironment{fmpage}[1]
        {\begin{lrbox}{\fmbox}\begin{minipage}{#1}}
        {\end{minipage}\end{lrbox}\fbox{\usebox{\fmbox}}}

\newcounter{algocnt}
\newenvironment{algolist}[1]{%
    \begin{list}{\thealgocnt}
    {\parsep 0in\usecounter{algocnt}\setcounter{algocnt}{0}\renewcommand{\thealgocnt}{{#1}\arabic{algocnt}}
    \setlength{\rightmargin}{0in}
    \settowidth{\leftmargin}{{#1}999}\addtolength{\leftmargin}{\labelsep}}}{\end{list}}
\newcommand{\algotab}[0]{\hspace*{1.2\labelsep}}
\newcommand{\mt}[0]{\algotab\ }

\newcommand{\setto}{\leftarrow}

\newcommand{\msd}{d}
\newcommand{\qexp}{r}

\newcommand{\algoinc}{\textsf{FP-Increment}}
\newcommand{\accuracy}{A}
\newcommand{\accbound}{B}
\newcommand{\pcnt}[2]{p_{#1}(#2)}

\begin{document}

\begin{titlepage}
\title{Approximate counting with a floating-point counter}

\author{Mikl\'os Cs\H{u}r\"os\thanks{
	Department of Computer Science and Operations Research,
	Universit\'e de Montr\'eal. E-mail:
	csuros AT iro.umontreal.ca}}

\maketitle

\begin{abstract}
Memory becomes a limiting factor in contemporary
applications, such as analyses of the Webgraph and molecular
sequences, when many objects need to be counted
simultaneously. Robert Morris [{\em Communications of the
ACM}, 21:840--842, 1978] proposed a probabilistic technique
for approximate counting that is extremely space-efficient.
The basic idea is to increment a counter containing the
value $X$ with probability $2^{-X}$. As a result, the
counter contains an approximation of $\lg n$ after~$n$
probabilistic updates stored in $\lg\lg n$ bits. Here we
revisit the original idea of Morris, and introduce a binary
floating-point counter that uses a $d$-bit significand in
conjunction with a binary exponent. The counter yields a
simple formula for an unbiased estimation of $n$ with a
standard deviation of about $0.6\cdot n2^{-d/2}$, and uses
$d+\lg\lg n$ bits.

We analyze the floating-point counter's performance in a
general framework that applies to any probabilistic counter,
and derive practical formulas to assess its accuracy. 
\end{abstract}

\end{titlepage}

\section{Introduction}
An elementary information-theoretic argument
shows that $\bigl\lceil \lg (n+1)\bigr\rceil$
bits are necessary to represent 
integers between~0 and~$n$ 
($\lg$ denotes binary logarithm throughout the paper).
Counting some interesting objects in a data 
set thus takes logarithmic space. 
Certain applications
need to be more economical because 
they need to maintain many counters simultaneously while, say, 
tracking patterns in large data streams. 
Notable examples where memory becomes a limiting factor
include analyses of the Webgraph~\cite{Bechetti.Web.spam,Donato.webgraph}. 
Numerous bioinformatics studies also
require space-efficient solutions
when searching for recurrent motifs in 
protein and DNA sequences. 
These frequent sequence 
motifs are associated with mobile, structural, regulatory or 
other functional elements, and have been studied since the first molecular sequences 
became available~\cite{Karlin.symposium}. 
Some recent studies have concentrated on patterns involving
long oligonucleotides, i.e., ``words'' of length 16--40 
over the 4-letter DNA alphabet, revealing
potentially novel regulatory features~\cite{JonesPevzner.longmotifs,Pyknon},
and general characteristics of copying processes 
in genome evolution~\cite{spectrum.dPln,Yorke.duplication.counts}. 
Hashtable-based indexing techniques~\cite{SSAHA} 
used in homology search and genome assembly procedures also rely on
counting in order to identify repeating sequence patterns.   
In these applications, billions of counters need to be 
handled, making implementations difficult 
in mainstream computing environments.  
The need for many counters is aggravated by the fact
that the counted features often have heavy-tailed
frequency distributions~\cite{spectrum.dPln,Donato.webgraph,Yorke.duplication.counts},
and there is thus no ``typical'' size for  
individual counters that could guide the memory allocation at the outset. 
As a numerical example, 
consider a study~\cite{spectrum.dPln} of the 16-mer distribution in the human genome sequence,
which has a length surpassing three billion. 
More than four billion ($4^{16}$)
different words 
need to be counted, 
and the counter values span more than sixteen binary magnitudes
even though the average 16-mer occurs only once or twice. 

One way to greatly reduce memory usage 
is to relax the requirement of exact counts. 
Namely, approximate counting to~$n$ is possible using
$\lg\lg n+O(1)$ bits with 
probabilistic techniques~\cite{Flajolet.analysis.Morris,Morris.counting}. 
The idea of probabilistic counting was introduced by Morris~\cite{Morris.counting}.
In the simplest case, a counter is initialized as $X=0$. 
The counter is incremented by one at the occurrence of an event with 
probability~$2^{-X}$. The counter is meant to track the magnitude of 
the true number of events. More precisely, 
after~$n$ events, the expected value 
of~$2^X$ is exactly~$(n+1)$.

A generalization of the binary Morris counter is the so-called {\em $q$-ary counter} with some 
$\qexp\ge 1$ and $q=2^{1/\qexp}$. In such a setup, the counter is incremented with probability~$q^{-X}$.
The actual event count is estimated as~$f(X)$, using the transformation
\[
f(x) = \frac{q^x-1}{q-1} = \frac{2^{x/\qexp}-1}{2^{1/\qexp}-1}.
\]
The function~$f$ yields an unbiased estimate, as 
$\EXP f(X) = n$ after~$n$ probabilistic updates.
The accuracy of a probabilistic counting method is characterized by 
the variance of the estimated count. For the $q$-ary counter, 
\begin{equation}\label{eq:var.q}
\Var f(X) = (q-1)\frac{n(n+1)}{2},
\end{equation}
which is approximately $\frac{\ln 2}{2 \qexp} n^2$ for large~$n$ and~$\qexp$.
The parameter~$\qexp$ governs the tradeoff between memory usage and 
accuracy. The counter stores~$X$ (with $n=f(X)$) 
using $\lg\qexp+\lg\lg n+o(1)$ bits; larger~$\qexp$ thus increases the accuracy at the expense of 
higher storage costs.

The main goal of this study is to introduce a novel algorithm for approximate counting.
Our {\em floating-point counter} is defined with the aid of a design parameter~$M=2^{\msd}$,
where~$d$ is a nonnegative integer. As we discuss later, $M$  
determines the tradeoff between memory usage and accuracy, analogously to parameter~$\qexp$ 
of the $q$-ary counter.
The procedure relies on a uniform
random bit generator~$\mathsf{RandomBit}()$. 
Algorithm $\algoinc$ below shows the incrementation procedure for a floating-point counter,
initialized with~$X=0$. 
Notice that the first~$M$ updates are deterministic.

\begin{center}
\begin{fmpage}{0.9\textwidth}
\begin{algolist}{}
\item[] $\algoinc(X)$ \hfill // {\em returns new value of $X$}
\item set $t\setto \lfloor X/M\rfloor$ \hfill // {\em bitwise right shift by $\msd$ positions}
\item \textbf{while} $t>0$ \textbf{do}
\item \mt\ \textbf{if} $\mathsf{RandomBit}()=1$ \textbf{then} \textbf{return} $X$
\item \mt\ set $t\setto t-1$
\item \textbf{return} $X+1$
\end{algolist}
\end{fmpage}
\end{center}

The counter value~$X=2^{\msd}\cdot t+u$, where~$u$ denotes 
the lower~$\msd$ bits, is used to estimate the actual count
$f(X)=(M+u)\cdot 2^t-M$. The counter thus stores~$X$ using $\msd + \lg\lg n +  o(1)$ bits. 
The estimate's standard deviation is $\frac{c}{\sqrt{M}} n$ where~$c$ fluctuates 
between about~$0.58$ and~$0.61$ asymptotically (see Corollary~\ref{cor:accuracy} for a 
precise characterization). Notice that a $q$-ary counter 
with~$r=M$ has asymptotically the same memory usage, and a standard deviation 
of about $\frac{0.59}{\sqrt r} n$ (see Eq.~\eqref{eq:var.q}).
Our algorithm thus has similar memory usage
and accuracy as $q$-ary counting. 
The floating-point counter is more 
advantageous in two aspects. First, the first~$M$ updates are deterministic, i.e.,
small values are exactly represented with convenience. Second, the counter 
can be implemented with a few elementary integer and bitwise operations, whereas a $q$-ary counter works with 
irrational probabilities. 
The random updates in the floating-point counter 
occur with exact integer powers~$2^{-i}$,  
and such random values can be generated using an average of~2 random bits.
Specifically, the~$\algoinc$ procedure
uses an expected number of $\Bigl(2-\frac{t}{2^t-1}\Bigr)$ calls to 
the random bit generator $\mathsf{RandomBit}()$.
In contrast, a $q$-ary counter needs a uniform random number 
in the range $(0,1)$ to produce a random event with probability~$2^{-X/r}$. 

The rest of the paper is organized as follows. 
In order to quantify the performance of floating-point counters, we found it fruitful
to develop a general analysis of probabilistic counting, which is of independent 
mathematical interest. Section~\ref{sec:general.thm} presents the main results
about the accuracy of probabilistic counting methods. 
First, Theorem~\ref{tm:f} shows that every 
probabilistic counting method 
has a unique unbiased estimator~$f$ with $\EXP f(X)=n$ after~$n$ 
probabilistic updates. Second, Theorem~\ref{tm:variance} shows that the accuracy of 
any such method is computable directly from the counter value. 
Finally, Theorem~\ref{tm:accuracy} gives relatively simple 
upper and lower bounds on the asymptotic accuracy 
of the unbiased estimator.
The proofs of the theorems are given in Section~\ref{sec:general.proofs}, 
which can be safely skipped on first reading.
Section~\ref{sec:fp} presents 
floating-point counters in detail, and mathematically characterizes their utility 
by relying on the results of Section~\ref{sec:general.thm}.
Section~\ref{sec:fp} further illustrates the theoretical analyses with simulation experiments
comparing $q$-ary and floating-point counters.


\section{Probabilistic counting}\label{sec:general.thm}
For a formal discussion of probabilistic counting, consider the 
Markov chain formed by the successive counter values.
\begin{definition}\label{def:counter}
A {\em counting chain} is 
a Markov chain $(X_n\colon n=0,1,\dotsc)$ with 
\begin{subequations}\label{eq:counter.general.def}
\begin{align}
X_0 & = 0;\\
\Probcmd{X_{n+1}=k+1}{X_n=k} & = q_k\\
\Probcmd{X_{n+1}=k}{X_n=k} & = 1-q_k,
\end{align}
\end{subequations}
where $0<q_k\le 1$ are the transition probabilities defining the counter.
\end{definition}

It is a classic result associated with probabilities in
pure-birth processes~\cite{KarlinTaylor} that the 
{\em $n$-step probabilities} $\pcnt{n}{k}=\PROB\{X_n=k\}$ are
computable by a simple recurrence (see
Equations~(\ref{eq:pcnt.rec.0}--\ref{eq:pcnt.rec.k}) later).
In case of probabilistic counting, we want to 
infer~$n$ from the value of~$X_n$ alone through a 
computable function~$f$. A given probabilistic counting method is 
defined by the transition probabilities and the 
function~$f$. As we will see later (Theorem~\ref{tm:f}), 
the transition probabilities determine a unique 
function~$f$ that gives an unbiased estimate of the update count~$n$.

\begin{definition}\label{def:f}
A function~$f\colon \mathbb{N}\mapsto\mathbb{N}$ is an 
{\em unbiased count estimator} for a given counting chain 
if and only if $\EXP f(X_n) = n$ holds for all~$n=0,1,\dotsc$. 
\end{definition}

In the upcoming discussions, we assume that 
the probabilistic counting method uses an unbiased count estimator~$f$. 
The merit of a given method is gauged by its accuracy, 
as defined below. 
\begin{definition}\label{def:accuracy}
The {\em accuracy} of the counter is the coefficient of variation
\[
\accuracy_n = \frac{\sqrt{\Var f(X_n)}}{\EXP f(X_n)}. 
\]
\end{definition}

The theorems below provide an analytical framework for evaluating 
probabilistic counters. Theorem~\ref{tm:f} shows that
the unbiased estimator is uniquely defined by a relatively simple expression
involving the transition probabilities. Theorem~\ref{tm:variance}
shows that the uncertainty of the estimate can be determined 
directly from the counter value. Theorem~\ref{tm:accuracy}
gives a practical bound on the asymptotic accuracy of the counter. 

\begin{theorem}\label{tm:f}
The function 
\begin{subequations}\label{eq:gen.estimator}
\begin{align}
f(0) & = 0 \label{eq:gen.estimator.0}\\
f(k) & = \frac{1}{q_0}+\frac{1}{q_1}+\dotsc+\frac{1}{q_{k-1}}. & \{k>0\} \label{eq:gen.estimator.k}
\end{align}
\end{subequations}
uniquely defines the 
unbiased count estimator~$f$ 
for any given set of transition probabilities~$(q_k\colon k=0,1,\dotsc)$. 
Thus, for any given counting chain, we can determine efficiently an unbiased estimator. 
\end{theorem}

Theorem~\ref{tm:f} confirms the intuition that the transition probabilities must be exponentially decreasing
in order to achieve storage on~$\lg\lg n + O(1)$ bits. Otherwise, with
subexponential $q_k^{-1}=2^{o(k)}$, 
one would have $f(k) = 2^{o(k)}$, leading to~$\lg n=o(k)$. 

The next definition provides a computable function for quantifying the uncertainty of~$f(X)$. 
\begin{definition}\label{def:g}
The {\em variance function} for a given counting chain is defined  
by
\begin{subequations}\label{eq:g.def}
\begin{align}
g(0) & = 0 \\
g(k) & = \frac{1-q_0}{q_0^2}+\frac{1-q_1}{q_1^2}+\dotsm+\frac{1-q_{k-1}}{q_{k-1}^2} & \{ k>0\}
\end{align}
\end{subequations}
\end{definition}

Theorem~\ref{tm:variance} below shows that the 
accuracy is computable directly from the counter value for any counting chain.  
The statement has a practical relevance (since count estimates can be coupled
with the variance function's value), 
and the variance function is used to evaluate the asymptotic accuracy 
of any counting chain (see Theorem~\ref{tm:accuracy}).

\begin{theorem}\label{tm:variance}
The variance function~$g$ of Definition~\ref{def:g}
provides an unbiased estimate for the variance of~$f$ from Theorem~\ref{tm:f}. 
Specifically, 
\begin{equation}
\Var f(X_n) =  \EXP g(X_n)
\end{equation}
holds for all~$n\ge 0$, 
where the moments refer to the space of $n$-step probabilities.
\end{theorem}

Theorem~\ref{tm:accuracy} is
the last main result of this section. 
The statement relates the asymptotics of
the variance function, the unbiased count estimator, and the counting chain's accuracy. 
\begin{theorem}\label{tm:accuracy}
Let~$\accuracy_n$ be the accuracy of Definition~\ref{def:accuracy}, and let
\begin{equation}
\accbound_k = \frac{\sqrt{g(k)}}{f(k)}. 
\end{equation}
Let~$\liminf_{k\to\infty} \accbound_k =\mu$. 
Suppose that $\limsup_{k\to\infty} \accbound_k=\lambda<1$ (and, thus, $\mu<1$). 
Then 
\begin{subequations}\label{eq:accuracy.bound}
\begin{align}
\limsup_{n\to\infty} \accuracy_n \le \frac{\lambda}{\sqrt{1-\lambda^2}}\\
\liminf_{n\to\infty} \accuracy_n \ge \frac{\mu}{\sqrt{1-\mu^2}}.
\end{align}
\end{subequations}
\end{theorem}

\paragraph{\textbf{Example}} Consider the case of a $q$-ary counter, where $q_i=q^{-i}$ with some $q>1$.  
Theorem~\ref{tm:f} automatically gives the unbiased count estimator 
\[
f(k) = \sum_{i=0}^{k-1} q_i^{-1} = \frac{q^{k}-1}{q-1}.  
\]
Theorem~\ref{tm:variance} yields the variance function
\[
g(k) = \sum_{i=0}^{k-1} \bigl(q_i^{-2} - q_i^{-1}\bigr) = \frac{q^{2k}-1}{q^{2}-1}-\frac{q^{k}-1}{q-1}.
\]
In order to use Theorem~\ref{tm:accuracy}, observe that 
\[
\lambda^2 = \lim_{k\to\infty} \frac{g(k)}{f^2(k)} = \frac{q-1}{q+1} < 1.
\]
Therefore, we obtain the known result \cite{Flajolet.analysis.Morris}
that $\lim_{n\to\infty} A_n^2 = \frac{\lambda^2}{1-\lambda^2} = \frac{q-1}2$.


\section{Proofs}\label{sec:general.proofs}

In what follows, we use the shorthand notation
\[
\pcnt{n}{k} = \PROB\{X_n=k\}
\]
for the $n$-step probabilities.
By~\eqref{eq:counter.general.def}, 
$\pcnt{0}{0}=1$, and the recurrences
\begin{subequations}\label{eq:pcnt.rec}
\begin{align}
\pcnt{n+1}{0} & =  (1-q_0)\pcnt{n}{0} \label{eq:pcnt.rec.0}\\*
\pcnt{n+1}{k} & =  (1-q_k)\pcnt{n}{k} + q_{k-1}\pcnt{n}{k-1}   & \{k>0\} \label{eq:pcnt.rec.k}
\end{align}
\end{subequations}
hold for all $n\ge0$.

\begin{lemma}\label{lm:estimator.unique}
The unbiased estimator is unique. 
\end{lemma}
\begin{proof}
Since~$\EXP f(0)=0$ is imposed, and $X_0=0$ with certainty, $f(0)=0$. 
For all~$n$, $\PROB\{X_n>n\}=0$, so
\[
\EXP f(X_n) = \sum_{k=0}^n \pcnt{n}{k} f(k) = n.
\]
Thus, for all $n>0$, 
\[
f(n) =  \frac{n-\sum_{k=0}^{n-1} \pcnt{n}{k} f(k)}{\pcnt{n}{n}} 
	= \frac{n-\sum_{k=0}^{n-1} \pcnt{n}{k} f(k)}{q_0 q_1 \dotsm q_{n-1}}, 
\]
which shows that~$f(n)$ is uniquely determined by $f(0), \dotsc, f(n-1)$
and the $n$-step probabilities.
\end{proof}

\begin{proof}[Proof of Theorem~\ref{tm:f}]
Define the durations $L_k(n) = \sum_{i=0}^{n-1} \{ X_i=k\}$, 
i.e., the number of times $X_i=k$ for $i<n$. 
Define also $L_k = \lim_{n\to\infty} L_k(n) = \sum_{i=0}^\infty \{X_i=k\}$.
Clearly, $\EXP L_k = 1/q_k$. 
By the linearity of expectations, 
\begin{align*}
\EXP L_k & = \EXP L_k(n) + \EXP \sum_{i=n}^\infty \{ X_i = k\} \\*
	& = \EXP L_k(n) + \Expcmd{ \sum_{i=n}^\infty \{ X_i=k\}}{X_n\le k}\PROB\{X_n\le k\}\\*
	& = \EXP L_k(n) + \PROB\{X_n\le k\} \EXP L_k,
\end{align*}
where we used the memoryless property of the geometric distribution in the last step.
Consequently, 
\begin{equation}
\EXP L_k(n) = \frac{\PROB\{X_n>k\}}{q_k}.
\end{equation}
Now, 
\[
\EXP \sum_{k=0}^{\infty} L_k (n)  = \sum_{k=0}^{\infty}\PROB\{X_n>k\} \frac{1}{q_k}
 = \sum_{k=0}^{n} \pcnt{n}{k} \sum_{i=0}^{k-1} \frac1{q_i} 
 = \sum_{k=0}^{n} \pcnt{n}{k} f(k) 
 = \EXP f(X_n). 
\]
Since $\sum_{k=0}^{\infty} L_k (n)=n$,
we have $\EXP f(X_n) = n$ for all~$n$. 
By Lemma~\ref{lm:estimator.unique}, no other function~$f$ has the same property. 
\end{proof}

\begin{proof}[Proof of Theorem \ref{tm:variance}]
By~\eqref{eq:pcnt.rec}, for all $n\ge 0$, 
\begin{align*}
\EXP f^2(X_{n+1}) 
	& = \sum_{k=0}^{n+1} \pcnt{n+1}{k} f^2(k)\\
	& = \sum_{k=0}^{n} (1-q_k)\pcnt{n}{k} f^2(k) + \sum_{k=1}^{n+1} q_{k-1} \pcnt{n}{k-1} f^2(k)\\
	& = \EXP f^2(X_n) - \sum_{k=0}^{n} q_k \pcnt{n}{k} f^2(k) 
		+ \sum_{k=1}^{n+1} q_{k-1} \pcnt{n}{k-1} \bigl(f(k-1)+q_{k-1}^{-1}\bigr)^2\\
	& = \EXP f^2(X_n) + 2 \sum_{k=0}^n \pcnt{n}{k} f(k) + \sum_{k=0}^n \pcnt{n}{k} q_k^{-1}\\
	& = \EXP f^2(X_n) + 2 n + \sum_{k=0}^n \pcnt{n}{k} q_k^{-1}.
\end{align*}
Since $\Var f(X_n) = \EXP f^2(X_n)-\Bigl(\EXP f(X_n)\Bigr)^2 = \EXP f^2(X_n)-n^2$, 
\begin{equation}\label{eq:var.rec}
\Var f(X_{n+1}) = \Var f(X_n) + \sum_{k=0}^n \pcnt{n}{k} q_k^{-1} -1. 
\end{equation}
By~\eqref{eq:g.def} and~\eqref{eq:pcnt.rec},
\begin{align*}
\EXP g(X_{n+1}) 
	& = \sum_{k=0}^{n+1} \pcnt{n+1}{k} g(k)\\*
	& = \EXP g(X_n)
		- \sum_{k=0}^{n} q_k \pcnt{n}{k} g(k)
		+ \sum_{k=1}^{n+1} q_{k-1} \pcnt{n}{k-1} \Bigl(g(k-1)+\frac{1-q_{k-1}}{q_{k-1}^2}\Bigr)\\
	& = \EXP g(X_n)
		+ \sum_{k=0}^{n} \pcnt{n}{k} \frac{1-q_{k}}{q_k}. \\
	& = \EXP g(X_n) + \sum_{k=0}^n \pcnt{n}{k} q_k^{-1} -1.
\end{align*}
By~\eqref{eq:var.rec},
$\Var f(X_{n+1})-\Var f(X_n) = \EXP g(X_{n+1}) - \EXP g(X_n)$ 
holds for all $n\ge 0$.
Since $\Var f(X_0)=\EXP g(X_0)=0$, $\Var f(X_n) = \EXP g(X_n)$ holds for all $n$.
\end{proof}

\begin{proof}[Proof of Theorem \ref{tm:accuracy}]
Define
\[
W_n = \frac{\Var f(X_n)}{\EXP f^2(X_n)} 
	= \frac{\sum_{k=0}^{\infty} \pcnt{n}{k}\cdot g(k) }{\sum_{k=0}^{\infty} \pcnt{n}{k} \cdot f^2(k) }.
\]

Let~$\epsilon>0$ be an arbitrary threshold. 
By the definition of~$\lambda$, there exists~$K$ such that 
\[
\frac{g(k)}{f^2(k)} < (1+\epsilon)\lambda^2
\]
for all~$k>K$. Therefore,
\begin{align*}
W_n & = \frac{\sum_{k=0}^{K} \pcnt{n}{k} g(k) + \sum_{k>K} \pcnt{n}{k} \cdot g(k) }{\sum_{k=0}^{K} \pcnt{n}{k} f^2(k) 
	+ \sum_{k>K} \pcnt{n}{k} f^2(k)}\\*
	& <  \frac{\sum_{k=0}^{K} \pcnt{n}{k} g(k) + (1+\epsilon)\lambda^2 \sum_{k>K} \pcnt{n}{k} f^2(k)}{\sum_{k>K} \pcnt{n}{k} f^2(k)}\\*
	& =  (1+\epsilon)\lambda^2 + \frac{\sum_{k=0}^{K} \pcnt{n}{k} g(k) }{\sum_{k>K} \pcnt{n}{k} f^2(k)}.
\end{align*}
Since~$q_k>0$ for all~$k$, $\lim_{n\to\infty}  \pcnt{n}{k}=0$ for all~$k$. 
Consequently, $\lim_{n\to\infty} \sum_{k=0}^{K} \pcnt{n}{k} g(k) =0$. As 
$\lim_{n\to\infty} \sum_{k>K} \pcnt{n}{k} f^2(k) = \infty$,
there exists~$N$ such that 
\begin{equation}\label{eq:W.bound}
W_n < (1+2\epsilon)\lambda^2 \qquad \text{for all $n>N$}.
\end{equation}

Since $\Var f(X_n) = \EXP f^2(X_n) - \EXP^2 f(X_n)$, 
\[
W_n = \frac{\Var f(X_n)}{\Var f(x_n)+n^2}.
\]
By~\eqref{eq:W.bound}, 
$\frac{\Var f(X_{n})}{\Var f(X_{n})+n^2}  \le (1+2\epsilon)\lambda^2$
for all~$n>N$. So,
\begin{align*}
\frac{\Var f(X_n)}{n^2} & \le \frac{(1+2\epsilon)\lambda^2}{1-(1+2\epsilon)\lambda^2}\\*
& = \frac{\lambda^2}{1-\lambda^2}\biggl(1+\frac{2\epsilon}{1-(1+2\epsilon)\lambda^2}\biggr).
\end{align*}
Since~$\epsilon$ is arbitrarily small and $\lambda^2<1$, 
\[
\limsup_{n\to\infty} \frac{\Var f(X_n)}{n^2} \le \frac{\lambda^2}{1-\lambda^2}.
\]

The lower bound is proven analogously. Let $\epsilon>0$ be an arbitrary threshold. Let~$K$ be such that 
$\frac{g(k)}{f^2(k)} > (1-\epsilon)\mu^2$ for all $k>K$. So, 
\[
W_n > \frac{(1-\epsilon)\mu^2 \sum_{k>K} \pcnt{n}{k} f^2(k)}{\sum_{k=0}^{K} \pcnt{n}{k} f^2(k) 
	+ \sum_{k>K} \pcnt{n}{k} f^2(k)}.
\]
For~$n$ large enough, $W_n > (1-2\epsilon) \mu^2$ holds. Since $\epsilon$ is arbitrarily small, and 
$\mu^2\le\lambda^2<1$, 
\[
\liminf_{n\to\infty} \frac{\Var f(X_n)}{n^2} \ge \frac{\mu^2}{1-\mu^2}.
\] 
\end{proof}


\section{Floating-point counters}\label{sec:fp}
The counting chain for a
floating-point counter is defined using a design parameter~$M=2^{\msd}$
with some nonnegative integer~$\msd$:
\begin{subequations}\label{eq:transZ}
\begin{align}
\Probcmd{X_{n+1}=k+1}{X_n=k} & = 2^{-\lfloor k/M\rfloor};\\
\Probcmd{X_{n+1}=k}{X_n=k} & = 1-2^{-\lfloor k/M\rfloor}.
\end{align}
\end{subequations}

\begin{figure}
\centerline{\includegraphics[width=\textwidth]{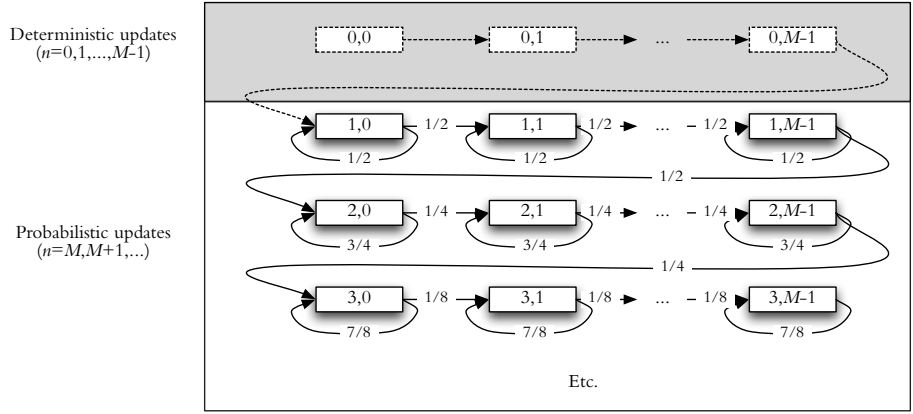}}
\caption{States of the counting Markov chain. Each state is labeled with a pair 
	$(t,u)$, where~$(u+M)$ are the most significant digits and~$t$ is the number of trailing zeros
	for the true count.}\label{fig:machine}
\end{figure}

Figure~\ref{fig:machine} illustrates the states of the floating-point counter.
The counter's designation becomes apparent 
from examining the binary representation of
the counter value~$k$.
Write $k=Mt+u$ with
\[
t = \lfloor k/M\rfloor \qquad\ 
u = k \bmod M;
\]
i.e., $u$ corresponds to the lower~${\msd}$ bits of~$k$, and~$t$ 
corresponds to the remaining upper bits. 
Theorem~\ref{tm:f} applies with~$q_k=2^{-\lfloor k/M\rfloor}$, leading to the following 
Corollary.
\begin{corollary}\label{cor:f}
The unbiased estimator 
for~$k=Mt+u$ is 
\begin{equation}\label{eq:fp.f}
f(k) = f(t,u) =(M+u)2^t-M.
\end{equation}
\end{corollary}
In other words, 
$(t,u)$ is essentially a floating-point 
representation of the true count~$n$,
where~$t$ is the exponent,
and~$u$ is a $d$-bit significand without the hidden bit for the leading `1.'

Theorem~\ref{tm:variance} yields the following Corollary.
\begin{corollary}\label{cor:variance}
The variance function for the floating-point counter is 
\begin{equation}\label{eq:fp.g}
g(k) = g(t,u) = \biggl(\frac{M}3+u\biggr)4^t - (M+u)2^t +\frac{2}{3}M. 
\end{equation}
\end{corollary}

Combining Corollaries~\ref{cor:f} and~\ref{cor:variance},  
we get the following bounds.
\begin{corollary}\label{cor:accuracy}
The accuracy of the floating-point counter is asymptotically bounded as
\begin{align*}
\limsup_{n\to\infty} \accuracy_n & \le \sqrt{\frac{3}{8M-3}}\\
\liminf_{n\to\infty} \accuracy_n & \ge \sqrt{\frac{1}{3M-1}}
\end{align*}
\end{corollary}
\begin{proof}
By Equations~\eqref{eq:fp.f} and~\eqref{eq:fp.g}, we have 
\[
\lim_{t\to\infty} \frac{g(t,u)}{f^2(t,u)} = \frac{\frac{M}3+u}{(M+u)^2}.
\]
Considering the extreme values at~$u=0$ and $u=M/3$, respectively: 
\begin{equation}
\mu^2 = \liminf_{k\to\infty} \frac{g(k)}{f^2(k)} = \frac13 M^{-1};\qquad
\lambda^2 = \limsup_{k\to\infty} \frac{g(k)}{f^2(k)} = \frac38 M^{-1}.
\end{equation}
Plugging these limits into Theorem~\ref{tm:accuracy} leads to the Corollary. 
\end{proof}

For large~$M=2^{\msd}$, the bounds of Corollary~\ref{cor:accuracy}
become 
\begin{align*}
\limsup_{n\to\infty} \accuracy_n & \lessapprox 2^{-\msd/2}\sqrt{3/8} \approx 0.612\cdot 2^{-\msd/2}\\
\liminf_{n\to\infty} \accuracy_n & \gtrapprox 2^{-\msd/2}\sqrt{1/3} \approx 0.577\cdot 2^{-\msd/2}.
\end{align*}

The accuracy is thus comparable to the accuracy of a~$q$-ary counter with~$q=2^{2^{-\msd}}$, which 
is approximately
$2^{-\msd/2} \sqrt{0.5\cdot \ln 2} \approx 0.589\cdot 2^{-\msd/2}$. 
The memory requirements of the two counters are equivalent: 
in order to count up to $n=f(k)$, $\lg k=d+\lg\lg n +o(1)$ bits are necessary. 

Figures~\ref{fig:trajectory} and~\ref{fig:estimate} compare 
the performance of the floating-point counters with 
equivalent base-$q$ counters in simulation experiments.  
The equivalence is manifest on 
Figure~\ref{fig:trajectory} that
illustrates the trajectories of the estimates by the different counters.
Figure~\ref{fig:estimate}
plots statistics about the estimates across multiple experiments: 
the estimators are clearly unbiased, and the two counters display 
the same accuracy.

\begin{figure}
\centerline{\includegraphics[width=\textwidth]{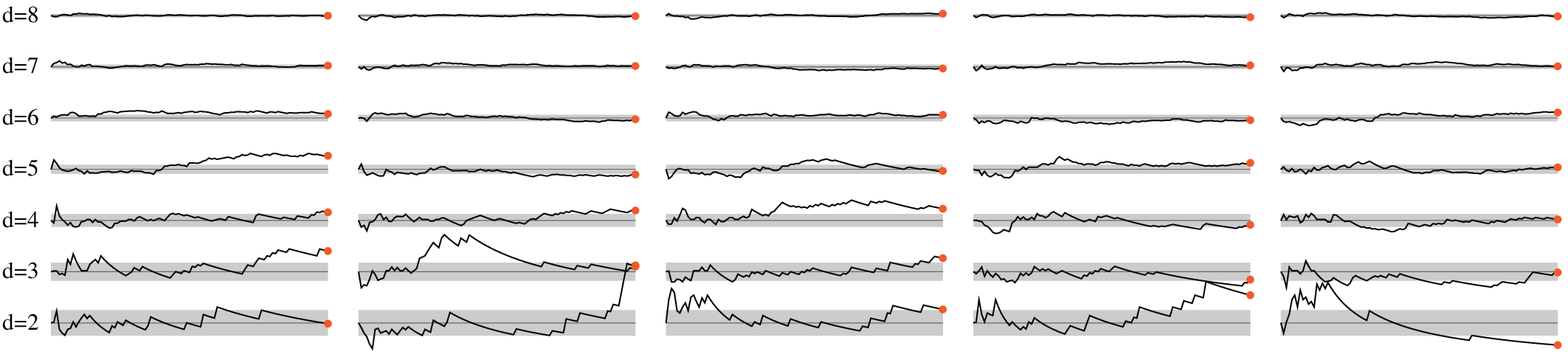}}
\centerline{\includegraphics[width=\textwidth]{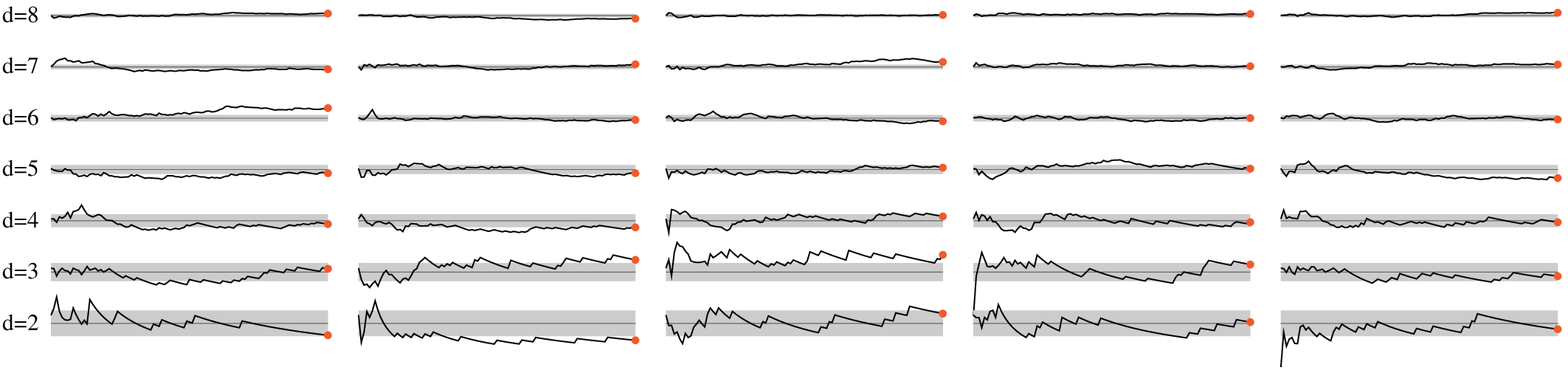}}
\caption{
Error trajectories for floating-point counters (\textbf{top}) and $q$-ary counters (\textbf{bottom}).
Each trajectory follows the the appropriate counting chain in a random simulated run. 
The lines trace the relative error $(f(X_n)-n)/n$ for 
floating-point counters with $\msd$-bit 
mantissa, and comparable $q$-ary counters with $q=2^{1/r}$
where $r=2^{\msd}$.
The shaded areas indicate a relative error of $\pm0.59\cdot 2^{-\msd/2}$.
The dots at the end of the trajectories denote the final value for $n=100000$. 
}\label{fig:trajectory}
\end{figure}

%

\begin{figure}
\centerline{\includegraphics[width=\textwidth]{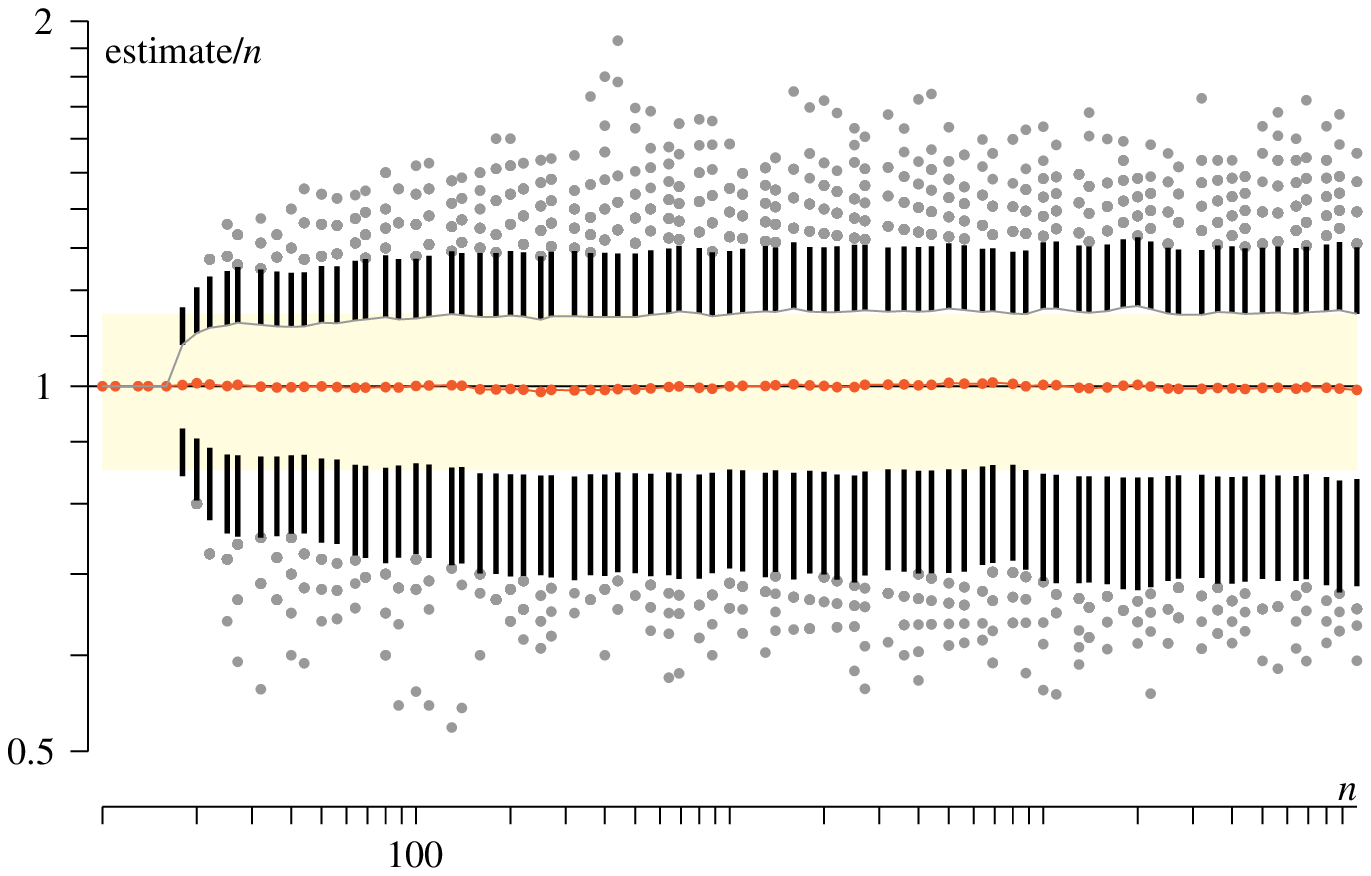}}
\centerline{\includegraphics[width=\textwidth]{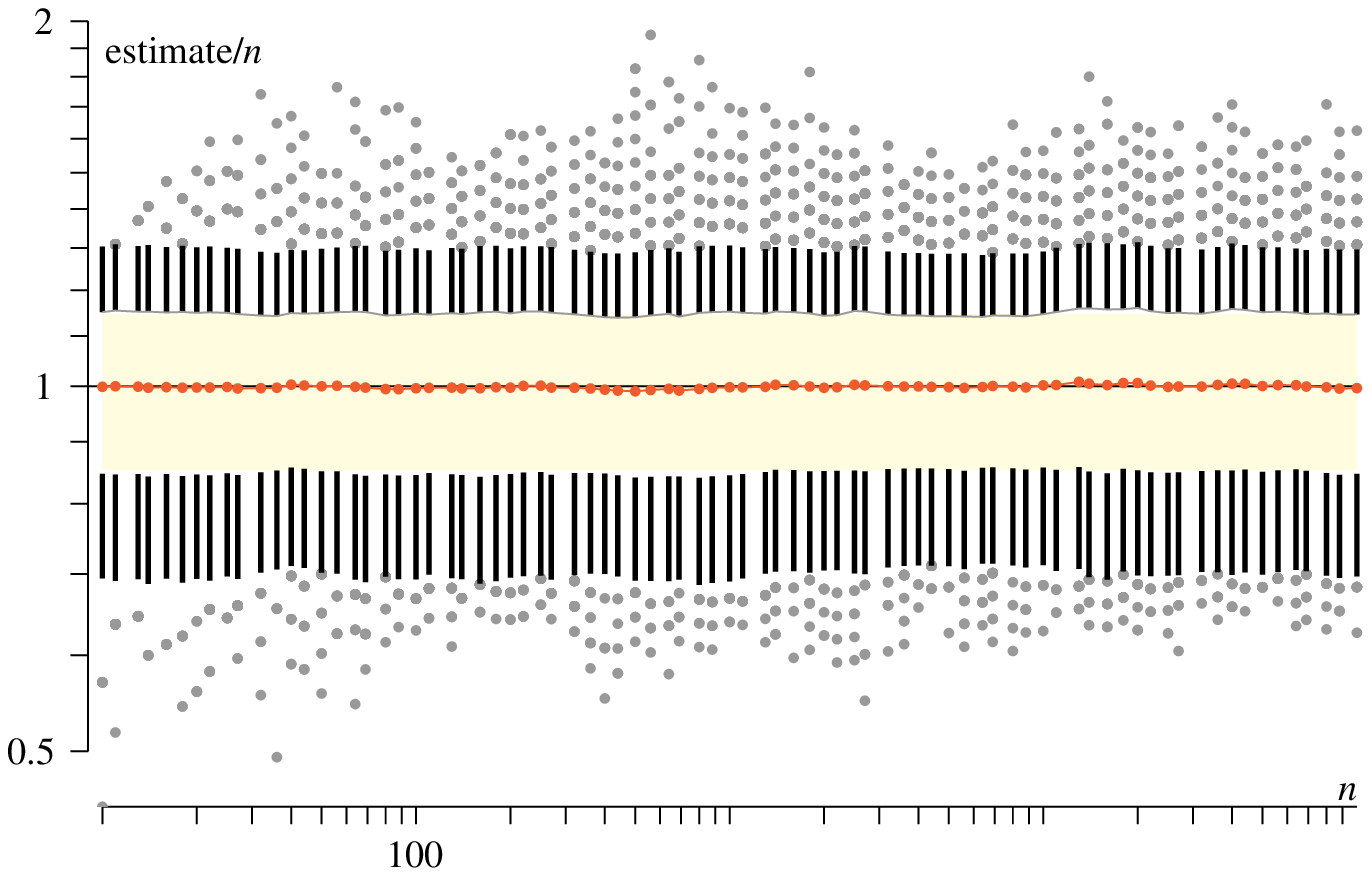}}
\caption{Distribution of the estimates for a floating-point counter (\textbf{top}) and 
a comparable $q$-ary counter (\textbf{bottom}).
Each plot depicts the result of 1000 experiments, in which a floating-point counter 
with $\msd=4$-bit mantissa, and a $q$-ary counter with $q=2^{1/16}$ were run until $n=100,000$. 
The dots in the middle follow the averages; the black segments 
depict the standard deviations 
(for each $\sigma$, they are of length~$\sigma$ spaced at $\sigma$ from the average),
and grey dots show outliers that differ by more than $\pm 2\sigma$ from the average. 
The shading highlights the asymptotic relative accuracy of the $q$-ary counter
($\approx 0.59\cdot 2^{-\msd/2}$).
}\label{fig:estimate}
\end{figure}


\section*{Acknowledgment}
I am very grateful to Philippe Flajolet for valuable suggestions 
on improving previous versions of the manuscript. 


\providecommand{\bysame}{\leavevmode\hbox to3em{\hrulefill}\thinspace}
\providecommand{\MR}{\relax\ifhmode\unskip\space\fi MR }
\providecommand{\MRhref}[2]{%
  \href{http://www.ams.org/mathscinet-getitem?mr=#1}{#2}
}
\providecommand{\href}[2]{#2}

\end{document}